\documentclass[runningheads]{llncs}
\usepackage{graphicx}
\usepackage{algorithm}
\usepackage[noend]{algpseudocode}
\usepackage{multirow}
\usepackage{amsmath, amssymb,pifont}
\usepackage{enumitem}
\usepackage[breaklinks=true]{hyperref}
\usepackage{breakcites}
\usepackage{wrapfig}
\usepackage{pdfpages}
\usepackage{fullpage}

\newenvironment{appendix-lemma}[1]{\vspace{0.1in}\noindent{\bf Lemma~#1~} \em }{\vspace{0.1in}}
\newenvironment{appendix-claim}[1]{\vspace{0.1in}\noindent{\bf Claim~#1~} \em }{\vspace{0.1in}}
\newenvironment{appendix-theorem}[1]{\vspace{0.1in}\noindent{\bf Theorem~#1~} \em }{\vspace{0.1in}}

\newcommand{\NP}{\sf NP}
\newcommand{\Poly}{\sf P}
\newcommand{\CL}{\mbox{{\sf CL}}}
\newcommand{\FPT}{\mbox{{\sf FPT}}}
\newcommand{\MAXEFM}{\mbox{{\sf MAXEFM}}}
\newcommand{\MAXRSM}{\mbox{{\sf MAXRSM}}}

\newcommand{\HR}{\sf HR}
\newcommand{\SM}{\sf SM}
\newcommand{\HRLQ}{\sf HRLQ}
\newcommand{\SMLQ}{\sf SMLQ}
\newcommand{\MANYONE}{\mbox{{\sf MANY}-{\sf ONE}}}
\newcommand{\ONEONE}{\mbox{{\sf ONE}-{\sf ONE}}}
\newcommand{\MANYONELQ}{\mbox{{\sf MANY}-{\sf ONE}-{\sf LQ}}}
\newcommand{\ONEONELQ}{\mbox{{\sf ONE}-{\sf ONE}-{\sf LQ}}}
\newcommand{\LQ}{\mbox{\sf LQ}}
\newcommand{\IS}{\mbox{{\sf IND}\text{-}{\sf SET}}}
\newcommand{\A}{\mathcal{A}}
\newcommand{\B}{\mathcal{B}}
\newcommand{\EE}{\mathcal{E}}

\newcommand{\alg}{{\sf ALG}}
\newcommand{\ALG}{{\sf EXTEND}}

\begin{document}
\title{Envy-freeness and Relaxed Stability for Lower-Quotas : A Parameterized Perspective}
\titlerunning{Envy-freeness and Relaxed Stability for LQs : A Parameterized Perspective}
%
\author{Girija Limaye} 
\authorrunning{G. Limaye}
%
\institute{Indian Institute of Technology Madras, India\\
\email{girija@cse.iitm.ac.in}{}{}
}

\maketitle              
\begin{abstract}
	We consider the problem of assigning agents to resources under the two-sided preference list model
	where resources specify an upper-quota and a lower-quota, that is, respectively
	the maximum and minimum number of agents that can be assigned to it.
	Different notions of optimality including {\em envy-freeness} and {\em relaxed stability}
	are investigated for this setting and the goal is to compute a largest
	{\em optimal} matching.
	Krishnaa~et~al.~\cite{DBLP:SAGT20} show that in this setting,
	the problem of computing a maximum size envy-free matching ($\MAXEFM$)
	or a maximum size relaxed stable matching ($\MAXRSM$)
	is not approximable within a certain constant factor unless $\Poly = \NP$.
	This work is the first investigation of parameterized complexity of $\MAXEFM$ and $\MAXRSM$.
	We show that $\MAXEFM$ is $W[1]$-hard and 
	$\MAXRSM$ is para-$\NP$-hard when parameterized on several natural parameters derived from the instance.
	We present kernelization results and $\FPT$ algorithms for both problems 
	parameterized on other relevant parameters.

\end{abstract}

\section{Introduction}\label{sec:intro}

Computing optimal matchings under the two-sided preference list setting is a well-studied problem.
In this setting, we are given a set of agents and a set of resources
such that
every agent and each resource {\em ranks} a subset from the other side 
and each resource has a positive upper-quota associated with it.
A {\em matching} in this setting is an allocation such that 
an agent is assigned to at most one resource and a resource is assigned to at most
upper-quota many agents.
Several practical applications naturally impose a {\em demand} on the number of 
agents matched to a resource. For instance, 
a hospital (that is, a resource) may require a minimum number of resident doctors (that is, agents)
to operate smoothly.
In such applications, a resource specifies a {\em lower-quota} in addition
to the upper-quota~\cite{DBLP:conf/wine/BoehmerH20, Yokoi20, HIM16, DBLP:SAGT20}.
A matching is {\em feasible} if every resource is matched to the number of agents
that is at least its lower-quota.

Formally, we model this as a bipartite graph $G = (V = \A \cup \B, E)$
where the bipartition consists of the set of agents ($\A$) and the of resources ($\B$).
Let $n = |V|, m = |E|$.
An edge $(a,b) \in E$ indicates that $a \in \A$ and $b \in \B$ are mutually acceptable.
Each vertex $v \in \A \cup \B$ ranks its neighbours in a strict order, denoted as the {\em preference list} of $v$.
If vertex $u$ prefers $v$ over $w$, we denote it by $v \succ_u w$.
Every resource $b$ has a positive upper-quota $q^+(b)$ and a non-negative
lower-quota $q^-(b)$
associated with it such that $q^-(b) \leq q^+(b)$.
If $q^-(b) > 0$, we denote the resource $b$ as an $\LQ$ resource otherwise it is a non-$\LQ$ resource.

A $\MANYONELQ$ instance $G$ is a bipartite graph and the associated quotas and preference lists.
$\ONEONELQ$ is a special case of $\MANYONELQ$ wherein every resource has an upper-quota of 1
(and therefore, a lower-quota of 0 or 1).
Given a $\MANYONELQ$ instance, a standard technique of {\em cloning}~\cite{GI89, DBLP:journals/algorithmica/MnichS20} can be used
to get the corresponding $\ONEONELQ$ instance.
If no resource is an $\LQ$ resource, then we denote the instance as $\MANYONE$ or $\ONEONE$ 
(depending on the upper-quota values).
In literature, $\MANYONE$ and $\ONEONE$ instances are respectively known as the Hospital-Residents $(\HR$) instance and the 
Stable Marriage ($\SM$) instance~\cite{GI89, GS62}.
When lower-quotas are present, these instances are known as $\HRLQ$ and $\SMLQ$ respectively~\cite{DBLP:SAGT20, Yokoi20}.

In this work, we present our results for the $\ONEONELQ$ setting.
(See Section~\ref{sec:contrib} for the details.)
A matching $M$ in this setting is a subset of edges representing an allocation of agents to resources
such that an agent is matched to at most one resource
and a resource is matched to at most one agent.
Note that an agent is matched to a resource if and only if the resource is matched to the agent.
Let $G$ be a $\ONEONELQ$ instance and $M$ be a matching in $G$.
Then $M(a)$ denotes the resource matched to agent $a$ and $M(b)$ denotes the agent matched to resource $b$.
If a vertex $v \in \A \cup \B$ 
is unmatched in $M$ then we denote it as $M(v) = \bot$.
A vertex prefers being matched to one of its neighbours over being unmatched, that is, $\bot$ is the least-preferred
choice for every vertex $v$.

An $\LQ$ resource $b$ is {\em deficient} if $b$ is unmatched in $M$.
A matching $M$ is {\em feasible} if no resource is deficient in $M$, otherwise it is {\em infeasible}.
In a $\ONEONE$ instance, all matchings are feasible.
We assume that the given $\ONEONELQ$ instance admits a feasible matching.
The notion of matching does not consider the preferences of agents and resources.
An instance may admit several feasible matchings and the goal is to compute a feasible matching
that is optimal with respect to the preferences.

\noindent{\bf Optimal feasible matchings. }
Stability is the de-facto notion of optimality for the $\ONEONE$ setting and is defined as follows.
\begin{definition}[Stability in the $\ONEONE$ setting]
	Let $M$ be a matching. A pair $(a, b) \in E \setminus M$ is a blocking pair w.r.t. $M$ if $b \succ_a M(a)$ and $a \succ_b M(b)$. A matching $M$ is {\em stable} if there is no blocking pair w.r.t. $M$.
\end{definition}

Recall that in the above definition, an unmatched vertex $v$ has $M(v) = \bot$ and it is the least-preferred choice for $v$.
A $\ONEONE$ instance always admits a stable matching and it can be efficiently computed by the well-known Gale and Shapley algorithm~\cite{GS62,GI89}.
However, a $\ONEONELQ$ instance may not admit a stable matching that is simultaneously feasible.
We recall the instance presented by Krishnaa~et~al.~\cite{DBLP:SAGT20} as Fig.~~\ref{fig:no_stb_feas}.
The instance in Fig.~\ref{fig:no_stb_feas} admits two feasible matchings: $M_1 = \{(a_1, b_2)\}$ and $M_2 = \{(a_1, b_2), (a_2, b_1)\}$ but neither of them is stable since $(a_1, b_1)$ forms a blocking pair.

\begin{figure}[!ht]
\begin{center}
    \begin{minipage}{0.2\textwidth}
    \begin{align*}
        a_1 &: b_1, b_2\\
        a_2 &: b_1
    \end{align*}
    \end{minipage}%
    \begin{minipage}{0.2\textwidth}
    \begin{align*}
        \text{[0,1]}\text{ } {b_1} &: a_1, a_2\\
        \text{[1,1]}\text{ } {b_2} &: a_1
    \end{align*}
    \end{minipage}%
	\caption{A $\ONEONELQ$ instance with two agents $\{a_1,a_2\}$ and two resources $\{b_1,b_2\}$. Quotas are denoted as [$q^-(b)$, $q^+(b)$] pair and the comma-separated list denotes the preference list of the corresponding vertex.
	}
\label{fig:no_stb_feas}
\end{center}
\end{figure}

When the stable matching $M$ is not feasible, there is at least one $\LQ$ resource that is deficient in $M$.
By the Rural Hospitals Theorem~\cite{Roth86}, the same set of agents and resources are matched in every stable matching.
Therefore, if an $\LQ$ resource is unmatched in a stable matching then it is unmatched (and hence deficient) 
in all stable matchings.
Using this fact, the {\em deficiency} of the instance is defined as follows.

\begin{definition}[Deficiency]\label{def:deficiency}
	Let $G$ be a $\ONEONELQ$ instance and $M$ be a stable matching in $G$. Then deficiency of $G$ is the number of $\LQ$ resources that are deficient in $M$.
\end{definition}

Thus, the deficiency is invariant of the stable matching.
Given that a $\ONEONELQ$ instance may not admit a stable, feasible matching,
several optimality notions~\cite{FITUY15, Yokoi20, NN17, DBLP:SAGT20, HIM16}
are considered.
Envy-freeness~\cite{Yokoi20, FITUY15, DBLP:SAGT20}, a relaxation of stability is defined as follows.

\begin{definition}[Envy-freeness]
Let $M$ be a matching. An agent $a$ has {\em justified envy} (here onwards called envy) towards a matched agent $a'$, where $M(a') = b$ and $(a,b) \in E$ if $b \succ_a M(a)$ and $a \succ_b a'$.
The pair $(a,a')$ is an envy pair w.r.t. $M$.  A matching $M$ is envy-free if there is no envy pair w.r.t. $M$.
\end{definition}

A $\ONEONELQ$ instance that fails to admit a stable, feasible matching may admit a feasible,
envy-free matching.
However, there exist simple $\ONEONELQ$ instances that do not admit a feasible, envy-free matching.
As noted in~\cite{DBLP:SAGT20}, for the instance shown in Fig.~\ref{fig:no_stb_feas}, the matching $M_1$ is the unique feasible, envy-free matching.
However, if $q^-(b_1) = q^-(b_2) = 1$ then the unique feasible matching $M_2$ is not envy-free because $a_1$ envies $a_2$.
Yokoi~\cite{Yokoi20} gives a characterization for the $\MANYONELQ$ instances that admit a feasible,
envy-free matching.

Krishnaa~et~al.~\cite{DBLP:SAGT20} propose a new notion of optimality, namely
relaxed stability~\cite{DBLP:SAGT20}.
Relaxed stability in the $\MANYONELQ$ instance is defined as follows.
In the $\MANYONELQ$ setting, let $M(b)$ denote the {\em set} of agents matched to resource $b$.
\begin{definition}[Relaxed stability in the $\MANYONELQ$ setting]
	Matching $M$ is {\em relaxed stable} if, for every resource $b$, at most $q^-(b)$ agents from $M(b)$ participate in blocking pairs and no unmatched agent participates in a blocking pair.
\end{definition}
In a $\ONEONELQ$ instance, the definition implies that if $b$ is a non-$\LQ$ resource then the agent $M(b)$ (if exists)
does not participate in a blocking pair and no unmatched agent participates in a blocking pair.
Equivalently, relaxed stability in the $\ONEONELQ$ setting is defined as follows.
\begin{definition}[Relaxed stability in the $\ONEONELQ$ setting]
	Matching $M$ is {\em relaxed stable} if, for every agent $a$ that forms a blocking pair, $a$ must be matched and (the resource) $M(a)$ is an $\LQ$ resource.
\end{definition}

\noindent 
For the instance shown in Fig.~\ref{fig:no_stb_feas}, if both resources have lower-quota of $1$,
then the unique feasible matching $M_2$ is relaxed stable, although it is neither stable nor envy-free.
In fact, every $\MANYONELQ$ instance always admits
a feasible, relaxed stable matching~\cite{DBLP:SAGT20}.
It is easy to observe that a stable matching is envy-free and relaxed stable.

\noindent{\bf Size of an optimal, feasible matching. }
The size of a matching plays an important role in practice.
A $\MANYONELQ$ instance may admit several feasible, envy-free matchings (when exists) or several
feasible, relaxed stable matchings of different sizes~\cite{DBLP:SAGT20}.
Krishnaa~et~al.~\cite{DBLP:SAGT20} show that
the size of a feasible, envy-free matching is at most the 
size of a stable matching (computed in the corresponding $\MANYONE$ instance by ignoring lower-quotas)
however a feasible, relaxed stable matching can even be larger than a stable matching
(even when a stable matching is feasible).
For the instance shown in Fig.~\ref{fig:no_stb_feas}, if $q^-(b_1) = 1$ and $q^-(b_2) = 0$ 
then the stable matching $M_s = \{(a_1, b_1)\}$ is feasible but matching $M_2$ is feasible, 
relaxed stable and is larger than the stable matching.
Therefore, it is natural to investigate the complexity of computing large
optimal, feasible matchings.

\noindent{\bf Complexity of computing a largest optimal, feasible matching. }
A maximum size feasible, stable matching (when exists) is efficiently computable in a $\MANYONELQ$ instance as follows --
first compute a stable matching $M$ in the corresponding $\MANYONE$ instance (by ignoring lower-quotas) 
using the Gale and Shapley Algorithm~\cite{GS62}.
Rural Hospitals Theorem~\cite{Roth86} implies that every resource $b$ is matched to the same {\em extend}
in every stable matching, that is, $|M(b)|$ is invariant of the stable matching $M$.
Thus, all stable matchings have the same size and either all are feasible or none of them is.

Krishnaa et al.~\cite{DBLP:SAGT20} investigate computational complexity of computing a maximum size feasible,
envy-free matching, when exists ($\MAXEFM$) and that of computing a maximum size
feasible, relaxed-stable matching ($\MAXRSM$) in a $\MANYONELQ$ instance.
They show that both these problems are $\NP$-hard and cannot
be approximated within $\frac{21}{19}-\epsilon, \epsilon>0$ unless $\Poly = \NP$, even under $\ONEONELQ$ setting
and other severe restrictions.
A polynomial time algorithm for a special case ($\CL$-restriction~\cite{DBLP:SAGT20, HIM16}) of $\MANYONELQ$ setting
and an approximation algorithm for the general case of the $\MAXEFM$ problem 
are presented in~\cite{DBLP:SAGT20} but the approximation guarantee is weak.
They also present an instance (Fig.~2~\cite{DBLP:SAGT20}) that admits two feasible, envy-free matchings 
such that the size difference is $O(|\A|)$.
Recall that unlike feasible, envy-free matchings,
a feasible, relaxed stable matching always exists in a $\MANYONELQ$ instance
and it could be larger than a stable matching.
This motivates the investigation of parameterized complexity of $\MAXEFM$ and $\MAXRSM$.
The following results from~\cite{DBLP:SAGT20} are used in this work.
\begin{lemma}[Lemma~2~\cite{DBLP:SAGT20}]\label{lem:stbl_feas}
	A stable matching, when feasible, is an optimal solution of $\MAXEFM$.
\end{lemma}

\begin{proposition}[from the proof of Lemma~2~\cite{DBLP:SAGT20}]\label{prop:p1}
	An agent unmatched in a stable matching is also unmatched in every envy-free matching.
\end{proposition}

\subsection{Our contributions}\label{sec:contrib}
This is the first work that investigates parameterized
complexity of computing largest feasible matchings that are either envy-free or 
relaxed stable in the lower-quota setting.
Our kernelization algorithms give
interesting insights into the structure of desired output matching.
Our $\FPT$ algorithms are simple to implement and their correctness proofs 
crucially use the fact (proved in this work)
that a {\em minimal feasible} optimal matching can be {\em extended} to a largest size one.

\noindent{\bf Parameters. }
$\MAXEFM$ is polynomial time solvable when there are no lower-quota resources.
Hence, it is natural to consider the case when the number of $\LQ$ resources
is small. We capture this by the parameter $q$ that denotes the number of $\LQ$ resources.
$\MAXEFM$ is also polynomial time solvable when a stable matching is feasible~\cite{DBLP:SAGT20}.
Therefore, we consider the case when the deficiency of the instance ($d$) or 
the number of deficient resources ($n_d$) is small.
We note that for a $\ONEONELQ$ instance, $d = n_d$ and for a $\MANYONELQ$ instance, $n_d \leq d$.

Let $\overline{\A}$ denote the union of the sets of agents acceptable to the $\LQ$ resources.
When a stable matching is infeasible, it is likely that there are few agents in $\overline{\A}$ 
and $\LQ$ resources have short preference lists.
We capture this by parameters $|\overline{\A}|$ and the length
of preference list of an $\LQ$ resource ($\ell_{LQ}$) respectively.
Also in practice, agents typically submit short preference lists.
Similar parameters are considered by Mnich~et~al.~\cite{DBLP:journals/algorithmica/MnichS20} 
for the optimization problem studied in~\cite{HIM16} for the $\MANYONELQ$ setting.
In~\cite{HIM16}, Hamada~et~al. investigate the problem
of computing a feasible matching with minimum number of blocking
pairs or blocking agents.
Mnich~et~al.~\cite{DBLP:journals/algorithmica/MnichS20} consider the length of preference list
and number of agents with lower-quota as parameters 
and present their results for the $\ONEONELQ$ setting.

We observe that $\MAXRSM$ is $\NP$-hard even when a stable matching is feasible 
(see Section~\ref{apsec:rsmcorr}).
As mentioned earlier, relaxed stable matchings can be even larger than a stable matching, 
even when stable matching is feasible.
By definition, a stable matching is {\em maximal}.
Therefore, the size of a 
feasible, relaxed stable matching 
is at most twice the size of a stable matching in the instance.
We consider the size of a stable matching ($s$) as a parameter.
Similar parameterization derived from the size of a matching 
is investigated in~\cite{AGRSZ18, DBLP:conf/fsttcs/Gupta0R0Z20}.
In~\cite{AGRSZ18}, the problem of computing a largest
stable matching in the $\SM$ setting
and in the {\em stable roommates} setting
\footnote{Stable roommates~\cite{Irving85} is a generalization of $\SM$ setting to non-bipartite graphs.}
with ties and incomplete lists are investigated.
They show that these problems admit
a polynomial kernel in the solution size and the size of a maximum matching respectively.
Gupta~et~al.~\cite{DBLP:conf/fsttcs/Gupta0R0Z20} investigate the problem of computing matchings
larger than a stable matching
by allowing minimum number of blocking pairs, parameterized by
the size of the output matching beyond the size
of a stable matching.
We remark that these problems are maximization problems and 
our results are in the same spirit.

A necessary condition for a pair of agents to form an envy pair
is that they both rank at least one common resource.
Let $t$ denote the number of non-$\LQ$ resources that are common in
preference list of a pair of agents.
We consider $t$ as a parameter. 
We remark that feasibility requirement may force an agent to get 
matched to a lower-preferred resource than a fixed resource, say $b$, 
thereby forbidding several agents
from getting matched to $b$ in order to achieve envy-freeness.
The parameter $t$ captures this dependency.

Table~\ref{tab:stru_params} summarizes the parameters used in this work.
\begin{table}[!ht]
	\centering{
	\begin{tabular}{|p{0.1\textwidth}|p{0.88\textwidth}|}\hline
		$q$ & Number of $\LQ$ resources \\\hline
		$\ell_{LQ}$ & Length of a longest preference list of an $\LQ$ resource \\\hline
		$d$ & Deficiency (Def.~\ref{def:deficiency}) of the instance w.r.t. a stable matching \\\hline
		$n_d$ & Number of deficient resources w.r.t. a stable matching \\\hline
		$|\overline{\A}|$ & Cardinality of the set of (unique) agents acceptable to $\LQ$ resources \\\hline
		$t$ & Number of common non-$\LQ$ resources appearing in the preference lists of a pair of agents \\\hline
		$s$ & Size of a stable matching in the instance \\\hline
	\end{tabular}
	}
	\caption {Summary of parameters}
	\label{tab:stru_params}
\end{table}

\noindent{\bf Our results. }
Recall that $\MAXEFM$ and $\MAXRSM$ respectively denote the problem
of computing a maximum size feasible, envy-free matching (when exists) and the problem
of computing a maximum size feasible, relaxed-stable matching in a $\MANYONELQ$ instance.
We present the following new results.
We show that $\MAXEFM$ is $W[1]$-hard on several interesting parameters.
\begin{theorem}\label{thm:hardness}
		$\MAXEFM$ is $W[1]$-hard when the parameter is $q$ or $n_d$ or $d$ 
		and is para-$\NP$-hard when parameter is $t$ even under $\ONEONELQ$ setting.
\end{theorem}

From the inapproximability result for $\MAXRSM$ in~\cite{DBLP:SAGT20}
we prove the following.
\begin{corollary}\label{corr:rsmparanphard}
$\MAXRSM$ is para-$\NP$-hard when parameterized by $d$ or $n_d$ even under $\ONEONELQ$ setting.
\end{corollary}

We present our algorithmic results for the $\ONEONELQ$ instances.
In Section~\ref{sec:disc} we discuss the applicability of these results for $\MANYONELQ$ instance.
We assume that the underlying bipartite graph of the input $\ONEONELQ$ instance 
admits a feasible, envy-free matching. 
We show that both problems admit a polynomial kernel
under relevant parameterization. 
\begin{theorem}\label{thm:kernel}
	$\MAXEFM$ admits a polynomial kernel in $(s,t)$.
	When stable matching is feasible, $\MAXRSM$ admits a polynomial kernel in $s$.
\end{theorem}

Although $\MAXEFM$ is $W[1]$-hard on $q$ alone, we show that $\MAXEFM$ admits an
$\FPT$ algorithm when parameterized on $q$ and $\ell_{LQ}$, and when parameterized 
on $|\overline{\A}|$. We show analogous $\FPT$ results for $\MAXRSM$.

\begin{theorem}\label{thm:fpt}
	$\MAXEFM$ and $\MAXRSM$ are $\FPT$ when parameters are ($q,\ell_{LQ})$ or $|\overline{\A}|$.
\end{theorem}

We summarize our results in Table~\ref{tab:summary}.

\begin{table}[!ht]
{\footnotesize
	\begin{tabular}{|p{0.1\textwidth}|p{0.3\textwidth}|p{0.22\textwidth}|p{0.22\textwidth}|}\hline
		\multirow{2}{*}{\bf Problem} & \multirow{2}{*}{\parbox{0.2\textwidth}{\bf Classical\\ complexity}} & \multicolumn{2}{c|}{\bf Parameterized results} \\\cline{3-4}
		& & \bf Hardness & \bf Polynomial kernel and $\FPT$ \\\hline
		$\MAXEFM$ & $\frac{21}{19}-\epsilon,\epsilon > 0$ inapproximability even for $\ONEONELQ$.~\cite{DBLP:SAGT20} In $\Poly$ when $q=0$ or when $d=0$ & $W[1]$-hard in each of $q, d, n_d$, 
		para-$\NP$-hard in $t$ even for $\ONEONELQ$ (Thm.~\ref{thm:hardness}) & $poly(s,t)$ kernel (Thm.~\ref{thm:kernel}), $\FPT$ in $(q,\ell_{LQ})$ and in $|\overline{\A}|$ (Thm.~\ref{thm:fpt}) \\\hline
		$\MAXRSM$ & $\frac{21}{19}-\epsilon,\epsilon > 0$ inapproximability even for $\ONEONELQ$ and when stable matching is feasible~\cite{DBLP:SAGT20} & para-$\NP$-hard in each of $d, n_d$ even for $\ONEONELQ$ (Cor.~\ref{corr:rsmparanphard}) & $poly(s)$ kernel when stable matching is feasible (Thm.~\ref{thm:kernel}), $\FPT$ in $(q,\ell_{LQ})$ and in $|\overline{\A}|$ (Thm.~\ref{thm:fpt}) \\\hline
	\end{tabular}
	\caption {Summary of results}
	\label{tab:summary}
	}
\end{table}

\noindent{\bf Other related work.}
In~\cite{MARX201125}, a $\MANYONE$ setting with couples parameterized by the number of couples is investigated.
Chen~et~al.~\cite{chen_et_al:LIPIcs:2018:9039} investigate parameterized complexity of 
variants of the stable roommates problem. 
In Bir{\'{o}}~et~al.~\cite{BFIM10}, colleges (resources) have lower-quotas that are either fulfilled or the colleges are closed. 
Under the setting studied in~\cite{BFIM10}, Boehmer~et~al.~\cite{DBLP:conf/wine/BoehmerH20} investigate the parameterized complexity of computing a stable matching (when exists). 
In~\cite{DBLP:journals/algorithmica/MarxS10}, 
the $\SM$ setting with ties and incomplete lists
parameterized by the number of ties and the maximum or total length of a tie is investigated.
Nasre and Nimbhorkar~\cite{NN17} investigate {\em popularity}
for $\MANYONELQ$ instances and show that a $\MANYONELQ$ instance always admits
a feasible, popular matching and 
that a maximum size feasible, popular matching is efficiently computable.
Empirical results comparing sizes of
feasible, envy-free and popular matchings 
are presented in~\cite{MNNR18}.
Structural properties of envy-free matchings under $\MANYONE$ setting are studied in~\cite{WR18}.

\vspace{0.1cm}
\noindent{\em Organization of the paper. }
We present our kernelization results in Section~\ref{sec:kernel} and $\FPT$ results in Section~\ref{sec:fptalgos}.
In Section~\ref{sec:hardness}, we present our hardness results.
In Section~\ref{sec:disc}, we discuss how our algorithmic results are applicable in $\MANYONELQ$ setting.

\section{Kernelization results}\label{sec:kernel}
In this section we present our kernelization results.
Given a $\ONEONELQ$ instance $G$ and an integer $k$,
we construct a kernel $G'$ such that $G$ admits a desired (feasible and either envy-free or relaxed stable) 
matching of size $k$ if and only if $G'$ admits a desired matching of size $k$ and 
$G'$ has polynomial size in the chosen parameter(s).
Recall that $s$ and $t$ respectively denote the size of a stable matching in the instance and
the number of non-$\LQ$ resources that are common in the preference list of a pair of agents.
Let $\ell(v)$ denote the length of the preference list of vertex $v$.

\subsection{A polynomial kernel for the $\MAXEFM$ problem}\label{sec:efmkernel}
We start by computing a stable matching $M_s$ in $G$. Then, $s = |M_s|$.
By Rural Hospitals Theorem~\cite{Roth86}, the set of agents unmatched in a stable matching
is invariant of the stable matching.
This fact and Proposition~\ref{prop:p1}
together imply that if $|M_s| < k$, we have a ``No'' instance.
A stable matching is envy-free, therefore if $|M_s| \geq k$ and $M_s$ is feasible, then 
we have a ``Yes'' instance. Otherwise, $|M_s| \geq k$ but $M_s$ is infeasible. 
We construct the graph $G'$ as follows.

\noindent {\bf Construction of the graph $G'$:}
Let $X$ be the vertex cover computed by picking matched vertices in $M_s$. Then, $|X| = 2s$. 
Let $X_B \subset X$, be the set of resources in $X$ and $X_A = X \backslash X_B$ be the set of agents in $X$. 
Then, $|X_B| = |X_A| = s$. 
Since, $M_s$ is maximal, $I = V \setminus X$ is an independent set.

\noindent {\bf Marking scheme: } 
We perform the following steps in the sequence given.
\begin{enumerate}
	\item For every $a \in X_A$ and every neighboring $\LQ$ resource $b$, mark the edge $(a,b)$.\label{st:stlq}
	\item For $b \in X_B$, 
		let $w(b)$ denote the minimum of $s+1$ and $\ell(b)$.
		Mark the highest-preferred $w(b)$ edges incident on $b$.
		\label{st:sthedges}
	\item Construct the set $C_a$ of resources corresponding to $a \in X_A$ as follows. The set $C_a$ consists of non-$\LQ$ resources that are common to the preference list of $a$ and some matched agent in $M_s$. That is,\\
		$C_a = \{b \in \B \mid$ $q^-(b) = 0$ and $ \exists a' \in X_A, a' \neq a$ and $b$ is in preference list of $a$ and $a'\}$.

		For every $a \in X_A$, $b \in C_a$, mark the edge $(a, b)$. \label{st:stcommonnlq}
	\item Amongst the unmarked edges (if any) incident on $a \in X_A$ mark the edge to the most-preferred resource $b$. \label{st:stextra}
\end{enumerate}

\noindent First we state the following observations about the marking scheme.
\begin{remark}\label{rem:r3}
	By step~\ref{st:stlq} of the marking scheme, for every $\LQ$ resource $b$,
	all the edges between $b$ and an agent in $X_A$ are marked.
\end{remark}
\begin{remark}\label{rem:r1}
	If an edge $(a,b)$ is not marked 
	for agent $a \in X_A$ or is marked in step~\ref{st:stextra} of the marking scheme 
	then resource $b$ is a non-$\LQ$ resource. This is because if $b$ were an $\LQ$ resource then the edge $(a,b)$
	would have been marked in step~\ref{st:stlq} of the above marking scheme.
	If $(a,b)$ is marked in step~\ref{st:stextra} then $b$ has exactly one neighboring agent $a$ in $X_A$, otherwise $(a,b)$ would have been marked in step~\ref{st:stcommonnlq}.
\end{remark}
\begin{remark}\label{rem:r2}
	If an edge $(a,b)$ is not marked 
	for agent $a \in X_A$ 
	then there exists a resource $b' \succ_a b$ such that $(a,b')$ was marked in step~\ref{st:stextra} of the marking scheme.
	By Remark~\ref{rem:r1}, $b'$ is a non-$\LQ$ resource and has exactly one neighbouring agent $a$ in $X_A$.
\end{remark}

Once the edges are marked using the above marking scheme, we let $G'$
be the graph spanned by the marked edges. Let $E'$ denote the edge set of $G'$.
Lemma~\ref{lem:3kernel}, Lemma~\ref{lem:3lemcorr} and Lemma~\ref{lem:3lemcorr2} prove that 
$G'$ is a polynomial kernel.

\begin{lemma}\label{lem:3kernel}
	The graph $G'$ has $O(poly(s, t))$-size.
\end{lemma}
\begin{proof}
	Recall that $q$ denotes the number of $\LQ$ resources in the instance.
	For each agent $a$ in $X_A$, let $q_a$ denote the number of $\LQ$ resources in its list. Then, $q_a \leq q$.
	Agent $a \in X_A$ has at most $q+(s-1)t+1$ edges marked.
	We remark that $q \leq s$, otherwise $G$ does not admit a feasible, envy-free matching.
	Each resource in $X_B$ has at most $s+1$ edges marked.
	Thus, $G'$ has at most $s(s+(s-1)t+1)+s(s+1)$ edges and no isolated vertices.
	Hence, $G'$ has $O(poly(s,t))$ size.
\end{proof}

\begin{lemma}\label{lem:3lemcorr}
If $G$ admits a feasible, envy-free matching of size $k$ then $G'$ admits a feasible,
envy-free matching of size $k$.
\end{lemma}
\begin{proof}
	Suppose $G$ admits a feasible, envy-free matching $M$ of size $k$.
	If $M \subseteq E'$ then $M$ itself is a feasible and envy-free matching in $G'$.
	Otherwise, an edge $(a,b) \in M$ is absent in $E'$.
	By Proposition~\ref{prop:p1}, $a \in X_A$ and therefore, by Remark~\ref{rem:r1},
	resource $b$ is a non-$\LQ$ resource.
	By Remark~\ref{rem:r2}, there exists a resource $b' \succ_a b$ such that $(a,b')$ was marked
	and $b'$ is a non-$\LQ$ resource and has exactly one neighbouring agent $a$ in $X_A$.

	We claim that $b'$ must be unmatched in $M$. To see this, first note that
	$b'$ has exactly one neighbouring agent $a$ in $X_A$ but $M(a) = b \neq b'$,
	therefore $M(b')$ (if exists) is not in $X_A$.
	If $M(b')$ exists and $M(b') \in I$ then it implies that an agent in $I$ is matched in $M$,
	a contradiction to Proposition~\ref{prop:p1}.
	Therefore, $b'$ must be unmatched in $M$.
	Now, we let $M = M \setminus \{(a,b)\} \cup \{(a,b')\}$.
	The resulting $M$ is a feasible matching since $b'$ was earlier unmatched and
	$b$ is a non-$\LQ$ resource. 
	
	We show that this change does not introduce envy.
	Suppose not. Then, let some agent, say $a'$, envies $a$ after this change.
	It implies that $a' \succ_{b'} a$.
	We show that such an agent $a'$ does not exist.
	As noted earlier, $b'$ has exactly one neighbouring agent $a$ in $X_A$,
	therefore $a' \notin X_A$. Thus, $a'$ must be in $I$.
	This implies that $b' \in X_B$, that is, $b'$ is matched in $M_s$.
	Since $b'$ has exactly one neighbouring agent $a$ in $X_A$, it must be the case
	that $(a,b') \in M_s$. This implies that $a \succ_{b'} a'$, otherwise $(a',b')$
	block the stability of $M_s$. 
	Therefore, $a'$ cannot be in $I$ as well. Thus, no agent envies $a$ after this change,
	that is, $M$ remains envy-free after this change.

	For every edge $(a,b) \in M \setminus E'$, there exists a unique resource $b'$
	with the properties mentioned above. We modify the matching $M$ by deleting the
	edge $(a,b)$ and adding the edge $(a,b')$. As shown earlier, this change results in a
	matching of size $k$ that is feasible and envy-free in $G'$.
\end{proof}

\begin{lemma}\label{lem:3lemcorr2}
If $G'$ admits a feasible, envy-free matching of size $k$ then $G$ admits a feasible,
envy-free matching of size $k$.
\end{lemma}
Before proving Lemma~\ref{lem:3lemcorr2}, we prove an important property that is used in the
proof of Lemma~\ref{lem:3lemcorr2}. 
Suppose $G'$ admits an envy-free matching $M'$.
Note that since some of the edges in $E$ are absent in $G'$, $M'$ need not be envy-free in $G$.
In particular, suppose that $a \succ_b M'(b)$ and $b \succ_a M'(a)$ and the edge $(a,b)$ 
is unmarked (that is, not present in $E'$). Then, $a$ does not envy $M'(b)$
in $G'$ but it envies $M'(b)$ in $G$ w.r.t. the matching $M'$.
In Claim~\ref{cl:cllem} we show that if such agent $a$ is in $I$ then $a$ must be unmatched in $M'$.

\begin{claim}\label{cl:cllem}
	Suppose $G'$ admits an envy-free matching $M'$ of size $k$.
	Let $a$ be an agent in $I$ who envies agent $a'$ in $G$ but not in $G'$ w.r.t. $M'$.
	Then, $a$ is unmatched in $M'$.
\end{claim}
\begin{proof}
	Suppose not.
	Then, let $b_1 = M'(a)$. Since $(a,b_1) \in M' \subseteq E' \subseteq E$,
	$b_1$ must be in $X_B$, that is, $b_1$ is matched in $M_s$.
	Let $a_1 = M_s(b_1)$. Then, $a_1 \succ_{b_1} a$,
	otherwise $(a,b_1)$ block the stability of $M_s$.
	We claim that the edge $(a_1,b_1) \in E'$.
	First observe that $a_1$'s rank in $b_1$'s preference list is at most $s$. 
	This is because, otherwise, since $|M_s| = s$, there is at least one 
	agent $\hat{a} \succ_{b_1} a_1$ that is unmatched in $M_s$.
	This contradicts the stability of $M_s$ since $(\hat{a}, b_1)$ form a blocking pair w.r.t. $M_s$.
	Therefore in step~\ref{st:sthedges}, $b_1$ marked the edge $(a_1,b_1)$, that is, $(a_1,b_1) \in E'$.

	Since, $M'$ is envy-free in $G'$ and $a_1 \succ_{b_1} a$ and $M'(b_1) = a$,
	$a_1$ must be matched to a resource, say $b_2$, in $M'$ such that $b_2 \succ_{a_1} b_1$.
	If $b_2 \in I$ then $(a_1, b_2)$ block the stability of $M_s$, therefore $b_2 \in X_B$, that is,
	$b_2$ is matched in $M_s$.
	Using the reasoning similar to that in case of $b_1$, we can extend the sequence such that
	at every resource, we must have an agent
	matched to that resource in $M_s$ such that the edge between them is present in $E'$
	and at every agent, we must have a resource
	matched to that agent in $M'$.
	Also, by the choice, every agent $a_1, a_2, \ldots$ along this sequence must be in $X_A$ since it
	is matched in $M_s$.
	Since $s=|X_A|$ is finite and the resource quotas are $1$,
	the sequence must reach agent $a$ such that for some resource $b_k$,
	$M_s(b_k) = a$. This implies that $a \in X_A$, a contradiction to the fact that $a \in I$.
	Therefore, $a$ must be unmatched in $M'$.
\end{proof}

Now we proceed to prove Lemma~\ref{lem:3lemcorr2}.

\begin{proof}[Proof of Lemma~\ref{lem:3lemcorr2}]
	Suppose $G'$ admits a feasible, envy-free matching $M'$ of size $k$.
	Since, $M' \subseteq E' \subseteq E$, the feasibility of $M'$ in $G$ follows.
	Suppose $M'$ is not envy-free in $G$.
	Then there exists an unmatched edge $(a,b) \in E \setminus E'$ such that
	$b$ is matched in $M'$ and $a$ envies $M'(b)$.
	Let $a' = M'(b)$. Then, by the claimed envy, $a \succ_b a'$.
	Note that $(a',b) \in M' \subseteq E'$.
	We consider the following two cases based on how the edge $(a',b)$ is marked during
	the marking scheme -- either $b$ marked it or $a'$ marked it.
	
	Suppose $b$ marked the edge $(a',b)$. This implies that $b$ must have marked the edge
	$(a,b)$ as well, since $a \succ_b a'$. This leads to a contradiction.
	Therefore, $a'$ must have marked the edge $(a',b)$, implying that $a' \in X_A$.
	First suppose that $a \in X_A$.
	Then by Remark~\ref{rem:r3}, $b$ is a non-$\LQ$ resource.
	However, $b$ is in the preference list of $a$ and $a'$ such that both $a$ and $a'$
	are in $X_A$, implying that $b \in C_{a}$. Therefore, the edge $(a,b)$ must be marked in step~\ref{st:stcommonnlq}
	of the marking scheme. This contradicts that $(a,b) \notin E'$.
	Therefore, $a \in I$.
	By Claim~\ref{cl:cllem}, $a$ is unmatched.
	Now, we modify $M'$ such that the resulting matching is feasible and envy-free in $G$.
	
	Let $U$ be the set of agents $a$ who envy some agent w.r.t. $M'$ in $G$.
	As shown earlier, such agent is in $I$ and is unmatched in $M'$.
	We run an agent-proposing Gale and Shapley stable matching algorithm~\cite{GS62} as follows.
	An agent newly added to $U$ starts proposing from the beginning of their preference list.
	We repeat the following until every unmatched agent in $U$ has exhausted its preference list.
	A resource accepts the proposal from a higher-preferred agent than its current match and 
	if an agent gets unmatched during this process for the first time, it is added to $U$.
	This process terminates after $O(m)$ iterations where $m = |E|$.
	Since agents propose, the matching remains feasible and its size remains $k$
	and a resource $b$ can only get matched to a better-preferred agent than $M'(b)$ before 
	the proposal sequence began (by the property of agent-proposing Gale and Shapley 
	algorithm).

	We claim that the resulting $M'$ is envy-free. Suppose not. Then, let $a_1$ envies
	$M'(b_1)$. If $a_1$ was in $U$ then $a_1$ must have proposed to $b_1$ and $b_1$ rejected $a_1$
	because $M'(b_1) \succ_{b_1} a_1$ at that time. Since, $M'(b_1)$ could have only improved later 
	on, $a_1$ does not envy $M'(b_1)$ at the termination. Next suppose that $a_1$ was never in $U$.
	It implies that $a_1$ is matched in $M'$ before the proposal sequence began.
	As shown earlier, (a matched agent) $a_1$ did not envy at this time. During the process, $M'(b_1)$ could
	have only improved, implying that $a_1$ does not envy $M'(b_1)$ at the termination.
	Therefore, the resulting matching $M'$ is envy-free in $G$.
\end{proof}

\subsection{A polynomial kernel for the $\MAXRSM$ problem}\label{sec:rsmkernel}
We present kernelization result for $\MAXRSM$ when a stable matching is feasible.
Recall that even when a stable matching is feasible,
$\MAXRSM$ is $\NP$-hard 
and a relaxed stable matching can be larger than a stable matching. 
We start by computing a 
stable matching $M_s$ in $G$. Then, $s = |M_s|$ and
by assumption, $M_s$ is feasible. 
Recall that a stable matching is relaxed stable,
therefore if $k \leq s$, then we have a ``Yes'' instance. 
Recall that due to maximality of $M_s$, the size of a relaxed stable matching in $G$ is at most $2s$.
Therefore, if $k > 2s$, then we have a ``No'' instance.
Otherwise, $s < k \leq 2s$.
Now we present our construction of the graph $G'$.

\noindent {\bf Construction of the graph $G'$:}
Let $X$, $X_B$, $X_A$ and $I$ be defined as in Section~\ref{sec:efmkernel}.
Recall that $\ell(v)$ denotes the length of the preference list of vertex $v$.

\noindent {\bf Marking scheme: } 
For a vertex $v$, let $w(v)$ denote the minimum of $2s+1$ and $\ell(v)$.
We perform the following steps. 
\begin{enumerate}
	\item For $a \in X_A$, mark the highest-preferred $w(a)$ edges incident on $a$. 
		\label{st:stzqr}
	\item For $b \in X_B$, mark the highest-preferred $w(b)$ edges incident on $b$. 
		\label{st:stzqh}
\end{enumerate}

Once the edges are marked using the above marking scheme, we let $G'$
be the graph spanned by the marked edges.
Let $E'$ denote the edge set of $G'$.
We note that $G'$ has $O(poly(s))$ size by observing
that $G'$ has no isolated vertices and 
each vertex $v \in X$ has at most $2s+1$ marked edges and $|X| = 2s$.
Now we proceed to show that $G'$ is a kernel.
We first observe that every edge $(a,b)$ matched in the stable matching $M_s$ is marked.

\begin{claim}\label{cl:cl1}
	If $(a,b) \in M_s$ then $(a,b)$ is marked.
\end{claim}
\begin{proof}
Suppose not. Since $a \in X_A$ and $(a,b)$ is not marked,
there are at least $2s+1$ resources neighbouring $a$ in $G$,
all higher-preferred over $b$. Since $|M_s| = s$, there is at least one
resource $b' \succ_a b$ that is unmatched in $M_s$.
Since, $M_s(a) = b$ and $b' \succ_a b$, edge $(a,b')$ forms a blocking pair 
	w.r.t. $M_s$, a contradiction. Therefore, edge $(a,b)$ is marked.
\end{proof}

Before showing that $G'$ is a kernel,
we prove an upper-bound on the size of a relaxed stable matching in $G'$ (Cor.~\ref{corr:gprime}).
Let $\mathcal{I}$ be a $\ONEONE$ instance and $N_s$ be a fixed stable matching in $\mathcal{I}$.
Let $\mathcal{I}^*$ be a $\ONEONE$ instance obtained from $\mathcal{I}$ by deleting a subset of edges $F$
and isolated vertices (if any) such that $F \cap N_s = \emptyset$.
Also, in $\mathcal{I}^*$ every vertex $v$ has
the same relative preference ordering as in $\mathcal{I}$.
Then, $N_s \subseteq E(\mathcal{I}^*)$ and $N_s$ is a matching in $\mathcal{I}^*$.
Moreover, $N_s$ is a stable matching in $\mathcal{I}^*$,
otherwise there exists a blocking pair $(a,b)$ w.r.t. $N_s$ in $\mathcal{I}^*$.
Since, $(a,b) \in E(\mathcal{I}^*) \subseteq E(\mathcal{I})$, $(a,b)$ also blocks the stability of $N_s$
in $\mathcal{I}$, a contradiction. 
It is easy to see that the above observation holds true even if $\mathcal{I}$ and $\mathcal{I^*}$ are $\ONEONELQ$ instances.

We note that the condition $F \cap N_s = \emptyset$ is crucial in the construction above.
By Claim~\ref{cl:cl1}, 
all edges in the fixed stable matching $M_s$ are marked. Therefore all edges in $M_s$ are present in the constructed instance $G'$.
Thus, the $\ONEONELQ$ instance $G'$ constructed above 
satisfies $F \cap M_s = \emptyset$
where $F$ is the set of unmarked edges.
Therefore, $M_s$ is stable in $G'$.
The fact that a stable matching is maximal and that $|M_s| = s$
imply the following corollary.

\begin{corollary}\label{corr:gprime}
	A relaxed stable matching in $G'$ has size at most $2s$.
\end{corollary}

Using Corollary~\ref{corr:gprime} and feasibility of the stable matching $M_s$ in $G$,
we show that $G'$ is a kernel.
\begin{lemma}\label{lem:rsmcorr}
	If $G'$ admits a feasible, relaxed stable matching $M'$ then $M'$ is feasible
	and relaxed stable in $G$.
\end{lemma}
\begin{proof}
	Since $M' \subseteq E' \subseteq E$, so feasibility of $M'$ in $G$ follows. 
	Suppose for the contradiction that $M'$ is not relaxed stable in $G$. 
	Then there exists blocking a pair $(a,b)$ such that $(a,b) \in E \setminus E'$
	that violates relaxed stability of $M'$ in $G$.
	By the claimed blocking pair, $b \succ_a M'(a)$ and $a \succ_b M'(b)$
	and either $M'(a) = \bot$ or $M'(a)$ is a non-$\LQ$ resource.

	First suppose that $a \in X_A$.
	Since $(a,b)$ is unmarked, there exist $2s+1$ edges incident on $a$,
	all higher-preferred over $b$ in $G'$.
	Since $|M'| \leq 2s$ (by Cor.~\ref{corr:gprime}), it implies that there is at least one resource $b' \succ_a b$
	that is unmatched in $M'$.
	Thus, $b' \succ_a M'(a)$ and hence $(a,b')$ form a blocking pair w.r.t. $M'$ in $G'$.
	Since $a$ is either unmatched in $M'$ or $M'(a)$ is a non-$\LQ$ resource,
	it contradicts relaxed stability of $M'$ in $G'$, a contradiction.
	Next, suppose that $a \in I$.
	Since $(a,b) \in E$, it implies that $b \in X_B$.
	By similar argument as above, there must exist an unmatched agent $a' \succ_b a \succ_b M'(b)$ such
	that $(a',b)$ is marked (that is, present in $E'$) and form a blocking pair w.r.t. $M'$ in $G'$.
	Since $a'$ is unmatched, it contradicts relaxed stability of $M'$ in $G'$.
	Therefore, the claimed blocking pair $(a,b)$ does not exist and hence,
	$M'$ is relaxed stable in $G$.
\end{proof}
	
\begin{lemma}\label{lem:rsmcorr2}
	If $G$ admits a feasible, relaxed stable matching $M$ then $M$ is feasible
	and relaxed stable matching of size $k$.
\end{lemma}
\begin{proof}
	If all the edges in $M$ are present in $G'$ then $M$ is a feasible, relaxed stable matching in $G'$.
	Otherwise, an edge $(a,b) \in M$ is absent in $E'$.
	First suppose that $b \in X_B$.
	Then there exist at least $2s+1$ agents neighboring $b$ in $G$, all
	higher-preferred over $a$.
	Since $|M| \leq 2s$, there is at least one unmatched agent $a' \succ_b a$.
	Then, (unmatched) $a'$ forms a blocking pair with $b$ w.r.t. $M$ in $G$, a contradiction to 
	relaxed stability of $M$ in $G$.
	Next, suppose that $b \in I$. It implies that $a \in X_A$
	and there exist at least $2s+1$ resources neighboring $a$ in $G$, all
	higher-preferred over $b$.
	Since $|M| \leq 2s$, there is at least one unmatched resource $b' \succ_a b$.
	Since $b \in I$ and $M_s$ is feasible in $G$, $b$ is a non-$\LQ$ resource.
	Then $(a,b')$ form a blocking pair such that $b=M(a)$ is a non-$\LQ$
	resource and $b'$ is unmatched and $b' \succ_a b = M(a)$.
	This contradicts the relaxed stability of $M$ in $G$.
	Thus, all edges in $M$ are present in $G'$. 
	This implies that $M$ is feasible and relaxed stable in $G'$.
\end{proof}

This establishes Theorem~\ref{thm:kernel}.

Our results in Section~\ref{sec:efmkernel} and in Section~\ref{sec:rsmkernel} are inspired 
by the kernelization result in~\cite{AGRSZ18}. 
The parameter $t$ is of theoretical interest given the para-$\NP$-hardness (Theorem~\ref{thm:hardness}) and provides interesting insights
in the kernelization process.
We remark that feasibility of a stable matching plays an important role
in designing the kernel for the $\MAXRSM$ problem.

\section{$\FPT$ algorithms}\label{sec:fptalgos}
In this section we present fixed parameter tractability of $\MAXEFM$
and $\MAXRSM$. Recall that $q$, $\ell_{LQ}$ and $|\overline{\A}|$ respectively denote the number of $\LQ$ resources,
length of a longest preference list of an $\LQ$ resource and
cardinality of the set of (unique) agents acceptable to $\LQ$ resources.

\subsection{$\FPT$ algorithm for the $\MAXEFM$ problem}
We begin by presenting
an algorithm ($\ALG$) that computes an 
{\em extension} of a given {\em minimal feasible} matching.
A matching is minimal feasible if it is feasible and deleting an edge results in an infeasible matching.
Matching $M'$ is an extension of $M$ if $M \subseteq M'$.
Given a minimal feasible matching $M$ in $G$,
our algorithm computes an extension $M'$ 
such that if $M'$ is envy-free in $G$ then $M'$ is a maximum size envy-free matching containing $M$ in $G$.
We use the notion of {\em threshold agent} which is similar to the {\em threshold resident} as defined in~\cite{MNNR18}.
Given a matching $M$, the threshold agent $t(b)$ of an unmatched resource $b$ is defined as follows.
\begin{definition}[Threshold agent]
	Threshold agent of an unmatched resource $b$ is
	the most-preferred agent $a$ in the preference list of $b$ such that $a$ is matched in $M$
	and $b \succ_{a} M(a)$, if such agent exists. Otherwise, we let a unique dummy agent $a_b$ at the end of $b$'s
preference list to be $t(b)$.
\end{definition}

The $\ALG$ algorithm (Algorithm~\ref{algo:extend}) takes a minimal feasible matching $M$ in $G$
as input. It starts by computing the threshold agents of unmatched resources.
Then it constructs a sub-graph $G'$ induced by the unmatched resources and
unmatched agents in $M$ that the resource prefers over its threshold agent.
For every resource $b$ in $G'$, $q^+(b)$ is set to 1.
For every vertex (agent and resource) in $G'$ its preference list is derived using the preference ordering 
same as that in $G$ restricted to its neighbors in $G'$. 
Then the algorithm computes the {\em agent-optimal} stable matching $M_s$ in $G'$.
The unique agent-optimal stable matching $M_s$ is obtained
by the agent-proposing Gale and Shapley algorithm
and has the following property -- if $N_s$ is a stable matching in
the instance then for every agent, either $M_s(a) = N_s(a)$ or $M_s(a) \succ_a N_s(a)$~\cite{GI89}.
Finally, it returns $M' = M_s \cup M$.

The algorithm $\ALG$ is similar to Algorithm~2~from~\cite{MNNR18}
with the following subtle differences. 
The initial matching in the algorithm from~\cite{MNNR18} is a feasible, envy-free matching in $G$
(computed using Yokoi's algorithm~\cite{Yokoi20})
but the input to $\ALG$ is a minimal feasible matching (not necessarily envy-free in $G$).
Also, $\ALG$ computes the agent-optimal stable matching $M_s$
unlike any stable matching computed in the algorithm in~\cite{MNNR18}.
We use $\ALG$ to design an $\FPT$ algorithm for $\MAXEFM$ and this choice of a stable matching
plays an important role.

\begin{algorithm}[!ht]
\begin{algorithmic}[1]
	\Statex {\bf Input:} $G$: $\ONEONELQ$ instance, $M$: a minimal feasible matching
	\Statex {\bf Output:} An extension $M'$ of $M$
	\State For every unmatched resource $b$, compute the threshold agent $t(b)$.
	\State Let $G'$ be the sub-graph of $G$, where
		$E(G') = \{(a,b) \mid b$ and $a$ are unmatched in $M$ and $a \succ_b t(b)\}$
	\State For every $b$ in $V(G')$, set $q^+(b) = 1$
	\State Each vertex in $G'$ has the preference list derived from $G$ restricted to its neighbors in $G'$
	\State Compute the agent-optimal stable matching $M_s$ in $G'$ and return $M' = M \cup M_s$
\end{algorithmic}
	\caption{Algorithm $\ALG$}
	\label{algo:extend}
\end{algorithm}

We note that the $G'$ constructed in the algorithm is a $\ONEONE$ instance (that is, there are no $\LQ$ resources).
It is straightforward to observe that Algorithm $\ALG$ takes linear time in the size of input (number of edges
in the underlying graph).
Before presenting our $\FPT$ algorithm, we 
prove an important property of $M'$
computed by $\ALG$ (Cor.~\ref{corr:extend}).

Let $M_1$ be an arbitrary envy-free matching in $G$ that contains the given mininal feasible matching $M$, 
that is, $M \subseteq M_1$.
We first observe that $M_1 \setminus M \subseteq E(G')$ -- suppose not, then
there exists an edge $(a',b') \in M_1 \setminus M$ such that $(a',b') \notin E(G')$.
This implies that $t(b') \succ_{b'} a'$, therefore $t(b')$ envies $a'$ w.r.t. $M_1$ in $G$, a contradiction.
Also, $M_1 \setminus M$ is envy-free in $G'$, otherwise $M_1$ is also not envy-free in $G$.
We observe that the stable matching $M_s$ is (trivially) feasible in $G'$,
hence by Lemma~\ref{lem:stbl_feas}, 
$M_s$ is a maximum size envy-free matching in $G'$.
Thus, $|M_1 \setminus M| \leq |M_s|$, thereby implying that 
$|M_1| = |M| + |M_1\setminus M| \leq |M| + |M_s| = |M'|$.
That is, the size of an arbitrary envy-free matching $M_1$ in $G$ that contains $M$ is
upper-bounded by the size of $M'$ computed by $\ALG$.
However, $M'$ need not be envy-free in $G$. Therefore, we obtain the following corollary.
\begin{corollary}\label{corr:extend}
	Given a minimal feasible matching $M$ in a $\ONEONELQ$ instance $G$, 
	Algorithm $\ALG$ computes an extension $M'$ of $M$
	such that if $M'$ is envy-free in $G$ then $M'$ is a maximum
	size envy-free matching containing $M$ in $G$.
\end{corollary}

\noindent
Now we present our $\FPT$ algorithm ($\alg_{efm}$) for the $\MAXEFM$ problem.
		\begin{itemize}
			\item [1] For each assignment $M_e$ of agents to $\LQ$ resources:
		\begin{itemize}
			\item [a.] 
				If $M_e$ is infeasible or if a matched agent in $M_e$ envies, then discard it. 
				Otherwise run Algorithm $\ALG$ to compute an extension $M_e'$ of $M_e$.
				If $M_e'$ is not envy-free in $G$, discard $M_e'$.
		\end{itemize}
	\item [2] Among all matchings $M_e'$ computed for different assignments $M_e$, output the largest size matching.
		\end{itemize}

\begin{lemma}\label{lem:fpt_efm}
	Algorithm $\alg_{efm}$ runs in $O(m\cdot \ell_{LQ}^q)$ and $O(m\cdot |\overline{\A}|!)$ time respectively for parameters $(q,\ell_{LQ})$ and $|\overline{\A}|$ and computes a maximum size feasible, envy-free matching in $G$.
\end{lemma}
\begin{proof}
	\noindent{\bf Running time.} For the parameters $(q,\ell_{LQ})$, the step~1 considers at most $\ell_{LQ}$
ways of assigning an agent to an $\LQ$ resource.
Thus, our algorithm considers at most $\ell_{LQ}^q$ many different assignments.
	For the parameter $|\overline{\A}|$, the step~1 considers at most $|\overline{\A}|!$ ways of assigning
an agent to an $\LQ$ resource.
Thus, our algorithm considers $|\overline{\A}|!$ many different assignments.
Checking whether an assignment is feasible and whether a matched agent envies and to extend it using Algorithm $\ALG$ takes linear time. 
	Thus the algorithm runs in $O(m\cdot \ell_{LQ}^q)$ and $O(m\cdot |\overline{\A}|!)$ time respectively for $(q,\ell_{LQ})$ and $|\overline{\A}|$.

\noindent{\bf Correctness.} Suppose $M^*$ is a maximum size feasible, envy-free matching in $G$.
Let $M^*_{LQ}$ be the matching $M^*$ restricted to the $\LQ$ resources.
	Our algorithm considers this assignment as $M_e = M^*_{LQ}$. 
	Note that $M_e$ is feasible. Next we show that $M_e$ is not discarded due to a matched agent that envies another (matched) agent in $M_e$.
Suppose there exists a matched agent $a$ in $M_e$ who envies another matched agent $a'$.
	Then, the envy pair also exists w.r.t. $M^*_{LQ}$ and therefore, w.r.t. $M^*$, a contradiction.
Hence, $M_e$ is not discarded and an extension $M_e'$ is computed by $\ALG$.

Next we show that $M_e'$ is envy-free in $G$ and hence not discarded.
	Suppose for the contradiction that $M_e'$ is not envy-free in $G$.
	Then there exist agents $a$ and $a'$ such that $a$ envies $a'$ in $M_e'$.
	Then $a'$ is matched in $M_e'$.
	First suppose that $a'$ is matched in $M_e$.
	By the choice of $M_e$, a matched agent in $M_e$ does not envy another matched agent in $M_e$,
	therefore $a$ is not matched in $M_e$.
	Thus, $a$ is either unmatched in $M_s$ (computed by $\ALG$) or $a$ is matched in $M_s$ such that $M_e(a') \succ_a M_s(a)$.

	Let $N = M^* \setminus M^*_{LQ}$.
	As discussed earlier, $N$ is an envy-free matching in $G'$.
	If $N$ is a stable matching in $G'$ computed in $\ALG$ then,
	by the agent-optimality of $M_s$, $M_s(a) = N(a)$ or $M_s(a) \succ_a N(a)$.
	In either case, $M^*_{LQ}(a') = M_e(a') \succ_a N(a)$.
	Thus, $a$ envies $a'$ in $M^*$, a contradiction.
	If $N$ is not a stable matching in $G'$ then since $N$ is envy-free,
	there is a blocking pair $(a_1,b_1)$
	such that $b_1$ is unmatched in $N$.

	For the analysis purpose, suppose we run a resource-proposing round~\cite{GI89} in $G'$
	such that all the unmatched resources in $N$ get a chance to propose to
	every agent in their list in $G'$. During this step, if a resource is unmatched for the first time,
	it also starts proposing from the beginning of its list.
	Let $N'$ be the matching computed by this step. Then, $N'$ is stable in $G'$.
	By the property of resource-proposing round, either $N'(a) = N(a)$ or $N'(a) \succ_a N(a)$
	for each agent $a$ that is unmatched in $M^*_{LQ}$.
	Now, due to the agent-optimality of $M_s$, we have $M_s(a) = N'(a)$ or $M_s(a) \succ_a N'(a)$.
	This implies that $M^*_{LQ}(a') = M_e(a') \succ_a N'(a)$, that is 
	$M^*_{LQ}(a') = M_e(a') \succ_a N(a)$, implying that $a$ envies $a'$ in $M^*$, a contradiction.

	Next, suppose that $a'$ is matched in $M_s$.
	Then, $a$ must be matched in $M_e$,
	otherwise $(a,M_s(a'))$ form a blocking pair w.r.t. $M_s$, a contradiction.
	Let $b = M_e(a)$ and $b' = M_s(a')$.
	By the claimed envy, $b' \succ_a b$ and $a \succ_{b'} a'$.
	Since the edge $(a',b') \in M_s \subseteq E(G')$, it implies that $a' \succ_{b'} t(b')$.
	By the definition of threshold agent, $t(b') = a$ or $t(b') \succ_{b'} a$.
	In either case, we get $a' \succ_{b'} a$, leading to a contradiction.
	Thus, $M_e'$ is envy-free in $G$ and hence it is not discarded.

	By Cor.~\ref{corr:extend}, $M_e'$ is a maximum size envy-free matching in $G$
	containing $M_e$. That is, $|M_e'| \geq |M^*|$.
	Let $\widehat{M}$ be the matching returned by our algorithm. Then, either $\widehat{M} = M_e'$
	or $|\widehat{M}| \geq |M_e'| \geq |M^*|$.
	By the choice of $M^*$, we conclude that $|\widehat{M}| = |M^*|$, that is,
	our algorithm outputs a maximum size envy-free matching in $G$.
\end{proof}

\subsection{$\FPT$ algorithm for the $\MAXRSM$ problem}
For the $\MAXRSM$ problem as well, our algorithm extends a given minimal feasible matching
and returns an extension of the maximum size.
We modify the algorithm $\alg_{efm}$ by replacing step 1a. appropriately and refer to this algorithm as $\alg_{rsm}$.

		\begin{itemize}
			\item [1] For each assignment $M_e$ of agents to $\LQ$ resources:
		\begin{itemize}
			\item [a.] If $M_e$ is not a feasible matching then discard it. 
			Otherwise construct a $\ONEONE$ instance $H$ as follows -- 
			prune the given $\ONEONELQ$ instance by removing the 
			agents and resources matched in $M_e$.
			Then compute the agent-optimal stable matching $M_s$ in $H$ and let $M_e' = M_e \cup M_s$.
		If $M_e'$ is not relaxed stable in $G$ then discard it.
		\end{itemize}
	\item [2] Among all matchings $M_e'$ computed for different assignments $M_e$, output the largest size matching.
		\end{itemize}

\begin{lemma}\label{lem:fpt_rsm}
	Algorithm $\alg_{rsm}$ runs in $O(m\cdot \ell_{LQ}^q)$ and $O(m\cdot |\overline{\A}|!)$ time respectively for parameters $(q,\ell_{LQ})$ and $|\overline{\A}|$ and computes a maximum size feasible, relaxed stable matching.
\end{lemma}
\begin{proof}
\noindent{\bf Running time.} The running time analysis is similar to that for $\alg_{efm}$.
Checking feasibility of $M_e$, constructing $H$ 
	and computing the agent-optimal stable matching takes linear time. 
	Thus the algorithm runs in $O(m\cdot \ell_{LQ}^q)$ and $O(m\cdot |\overline{\A}|!)$ time respectively for $(q,\ell_{LQ})$ and $|\overline{\A}|$.

\noindent{\bf Correctness. } 
	Our approach to prove the correctness of $\alg_{rsm}$
	is similar to the proof of correctness of $\alg_{efm}$
	but crucially depends on the notion of relaxed stability and the modified step 1a.
Suppose $M^*$ is a maximum size feasible, relaxed stable matching in $G$.
Let $M^*_{LQ}$ be the matching $M^*$ restricted to the $\LQ$ resources.
Our algorithm considers this assignment $M_e = M^*_{LQ}$
and since it is feasible the matching $M_e'$ is computed for this assignment.
	Let $N = M^* \setminus M^*_{LQ}$.
	We show that $N$ is stable in $H$ -- if not there is a blocking pair $(a,b)$
	such that $a$ is either unmatched or matched to a lower-preferred resource $b'$ in $H$.
	Since $H$ is a $\ONEONE$ instance, $b'$ is a non-$\LQ$ resource.
	This implies that $N$ violates relaxed stability and so does $M^*$, a contradiction.

We claim that the computed matching $M_e' = M_s \cup M_e$ is relaxed stable in $G$ --
suppose not. Then there exists an agent $a$ 
such that $a$ blocks relaxed stability of $M_e'$. Note that such $a$ is unmatched in $M_e = M^*_{LQ}$ and hence 
	each of $M^*(a) (= N(a))$ and $M_e'(a) (= M_s(a))$ is either $\bot$ or a non-$\LQ$ resource.
Since $M_s$ is stable in $H$ but $M_e'$ is not relaxed stable in $G$,
	$a$ must form a blocking pair with an $\LQ$ resource $b$ such that $M_e(b) = a'$
and $a' \prec_b a$ and $M_s(a) \prec_a b$.
But since $M_e = M^*_{LQ}$, $(a',b) \in M^*_{LQ} \subseteq M^*$. 
	By the agent-optimality of $M_s$ in $H$, we have either $M_s(a) = N(a)$ or
$M_s(a) \succ_a N(a)$.
	Thus, we get $M^*(a) = N(a) \prec_a b$ implying that $(a,b)$ forms a blocking pair w.r.t. $M^*$ as well. Since $a$ is either
unmatched in $M^*$ or $M^*(a)$
is a non-$\LQ$ resource, the blocking pair $(a,b)$ violates relaxed stability of $M^*$ in $G$ -- a contradiction.
\end{proof}

This establishes Theorem~\ref{thm:fpt}.

We remark that our $\FPT$ algorithms are simple and hence practically appealing
however the correctness proofs are non-trivial and crucially dependent on the optimality notion.

\section{Parameterized hardness results}\label{sec:hardness}
In this section, we prove Theorem~\ref{thm:hardness} and Corollary~\ref{corr:rsmparanphard}.

\subsection{Proof of Theorem~\ref{thm:hardness}}
We remark that the inapproximability result of $\MAXEFM$ problem 
(Theorem~1~\cite{DBLP:SAGT20}) does not immediately imply $W[1]$-hardness of $\MAXEFM$.
We present a reduction from Independent Set ($\IS$) problem.
Let $\langle G = (V,E), k\rangle$ be an instance of the
$\IS$ problem where $n = |V|$ and $m = |E|$. The goal in $\IS$
is to decide whether $G$ has an independent set of size  $k$ i.e. a subset of $k$ vertices that are pairwise non-adjacent. 
We construct a $\ONEONELQ$ instance $G' = (\A \cup \B, E')$ of the $\MAXEFM$ problem as follows.

\begin{figure}[!ht]
\begin{center}
\begin{minipage}{0.4\textwidth}
\begin{align*}
	1 \leq i \leq n,\ a_i &: B_i, X\\
	1 \leq j \leq m,\ a'_j &: B_{j1}, B_{j2}
\end{align*}
\end{minipage}%
\hfill
\begin{minipage}{0.5\textwidth}
\begin{align*}
	1 \leq i \leq n,\ 1 \leq u \leq q_i,\ [0,1]\ b_i^u &: a_i, \EE_i\\
	1 \leq j \leq k,\ [1,1]\ x_j &: a_1, a_2, \ldots, a_n
\end{align*}
\end{minipage}%
	\end{center}
    \caption{Reduced $\ONEONELQ$ instance $G'$ of $\MAXEFM$ from instance $\langle G, k \rangle$ of $\IS$.}
    \label{fig:reduction_is1}
\end{figure}

\noindent{\bf Reduction:} For every vertex $v_i \in V$, we have a vertex-agent $a_i \in \A$; for every edge $e_j \in E$,
we have an edge-agent $a_j' \in \A$. Thus $|\A| = m+n$. 
Let $E_i$ denote the set of edges incident on $v_i$ in $G$. 
Let $\EE_i$ denote the set of edge-agents corresponding to edges in $E_i$.
Let $q_i = |E_i| +1$ for all vertices $v_i \in V$. For every vertex $v_i \in V$, let $B_i = \{b_i^1, b_i^2, \ldots, b_i^{q_i}\}$ be the set of resources corresponding to $v_i$. Let $X = \{x_1, x_2, \ldots, x_k\}$ be also a set of $k$ resources. Every resource in set $B_i$ has zero lower-quota and an upper-quota equal to $1$. Every resource $x_i \in X$ has both lower and upper-quota equal to $1$. 

\noindent{\bf Preference lists:} The preferences of agents and resources are given in Fig.~\ref{fig:reduction_is1}. We fix an arbitrary ordering on sets $X$, $B_i$, $\EE_i$. A vertex-agent $a_i$ has the set $B_i$ followed by set $X$. An edge-agent $a_j'$ has two sets of resources (denoted by $B_{j1}$ and $B_{j2}$) corresponding to the end-points $v_{j1}, v_{j2}$ of the edge $e_j$. Every resource $b_i^u \in B_i$ has the vertex-agent $a_i$ followed by the edge-agents in $\EE_i$. Finally the resources in $X$ have in their preference lists all the $n$ vertex-agents.

\noindent {\bf Properties of the reduced instance. }
Recall that $q$ denotes the number of $\LQ$ resources in the instance.
In the reduced instance $G'$ there are $k$ $\LQ$ resources, hence $q = k$.
A stable matching in $G'$ is infeasible. This property is necessary otherwise 
by Lemma~\ref{lem:stbl_feas} 
the instance is trivially solvable in polynomial time.
Recall that $d$ and $n_d$ respectively denote the deficiency of the instance
and the number of deficient resources (w.r.t. a stable matching).
The deficiency of $G'$ is $k$, hence $G'$ has $d = n_d = k$.
A pair of vertex-agents does not have a common non-$\LQ$ resource
in their preference list.
An edge-agent $a_j'$ has at most $q_i$ non-$\LQ$ resources in its list that are common with
another edge-agent $a_k'$ or with a vertex-agent $a_i$.
Thus in $G'$, $t$ is same as the maximum value of $q_i$.
If every vertex $v_i$ in $G$ has a constant degree 
$w$ then $q_i = w+1$ for every vertex $v_i$, implying that $t$ is constant in $G'$.

\begin{lemma}\label{lem:lem_is1}
$G$ has an independent set of size $k$ if and only if $G'$ has a feasible, envy-free matching of size $m+n$.
\end{lemma}
\begin{proof}
	First suppose that $S \subseteq V$ is an independent set of size $k$ in graph $G$. We construct a
	feasible, envy-free matching in $G'$ which matches all the agents in $\A$. Let $T$ be the set of vertex-agents corresponding to the vertices $v_i \in S$ i.e. $T = \{a_i \mid v_i \in S\}$. Match $T$ with $X$ using Gale and Shapley stable matching algorithm~\cite{GS62}. Hence, the matching is feasible.
	Let $T'$ be the set of agents corresponding to vertices $v_i$ such that $v_i \notin S$ i.e. $T' = \{a_1, a_2, \ldots, a_n\} \setminus T$. Let $B'$ be the set of resources appearing in sets $B_i$ such that $v_i \in S$ i.e. $B' = \bigcup\limits_{i: v_i \in S}{B_i}$. Match $T' \cup \{a_1', \ldots, a_m'\}$ with $\B \setminus (B' \cup X)$ using Gale and Shapley stable matching algorithm.

	We now prove that the matching is envy-free. No pair of agents in $T$ form an envy pair because we computed a stable matching between $T$ and $X$. No pair of agents in $T' \cup \{a_1', \ldots, a_m'\}$ form an envy pair because we computed a stable matching between this set and $\B \setminus (B' \cup X)$. Since, all resources in $B'$ are forced to remain empty, no agent in set $T$ can envy an agent in set $\{a_1', \ldots, a_m'\}$. An agent in $T'$ is matched to a higher-preferred resource than any resource in $X$, hence such agent cannot envy any agent in $T$. Thus, the matching is envy-free.

We now prove that the matching size is $m+n$. Every vertex-agent $a_i$ is matched either with some resource in $X$ or some resource in $B_i$. Since, $S$ is an independent set, at least one end point of every edge is not in $S$. So for every edge $e_t=(v_{t1}, v_{t2})$, there is at least one resource in sets $B_{t1}$, $B_{t2}$ that can get matched with the edge-agent $a_t'$ without causing envy. Thus, every edge-agent is also matched. Thus, we have an envy-free matching of size $m+n$.

For the other direction, let us assume that $G$ does not have an independent set of size $k$. 
	Suppose for contradiction that there exists a feasible,
envy-free matching $M$ of size $m+n$ in $G'$. Due to the unit lower-quota of every $x_i \in X$, exactly $k$ vertex-agents must be matched to resources in $X$. Let $S \subseteq V$ be the set of vertices $v_i$ such that the corresponding vertex-agent $a_i$ is matched to some resource in $X$ in $M$, i.e. $S = \{v_i \mid M(a_i) \in X\}$. So, $|S| = k$. Since, $S$ is not an independent set, there exists at least two vertex-agents $a_x$ and $a_y$ matched to some resource in $X$ such that the edge $e_j = (v_x, v_y) \in E$. Due to the preference lists of the resources, all the resources in both $B_x$ and $B_y$ sets must remain empty in $M$ to ensure envy-freeness. This implies that the edge-agent $a_j'$ must be unmatched in $M$. This implies that $|M| < m+n$, a contradiction. This completes the proof of the lemma.
\end{proof}

$\IS$ is $\NP$-hard even for $w$-regular graphs for $w \geq 3$~\cite{FGHSJ}.
As noted earlier, the reduced instance $G'$ has constant $t=w+1$ if $G$ is $w$-regular
for some constant $w$.
Thus, Lemma~\ref{lem:lem_is1} implies that $\MAXEFM$ is $\NP$-hard even when $t$ is constant,
thereby implying that $\MAXEFM$ is para-$\NP$-hard in $t$.
$\IS$ is $W[1]$-hard when the parameter is solution size~\cite{DF12}. 
Let $k'$ be one of the parameters $q, d, n_d$ 
and consider $\MAXEFM$ parameterized by $k'$.
For the $\MAXEFM$ instance $G'$ constructed above,
if we let $k' = k$, then Lemma~\ref{lem:lem_is1} 
implies that $\MAXEFM$ is $W[1]$-hard in the chosen parameter.
This establishes Theorem~\ref{thm:hardness}.

\subsection{Proof of Corollary~\ref{corr:rsmparanphard}}\label{apsec:rsmcorr}
We show that the corollary follows from the reduction from~\cite{DBLP:SAGT20} that shows
hardness of approximation of the $\MAXRSM$ problem.
We note that the reduced instance in Fig.~4~\cite{DBLP:SAGT20} has upper-quotas equal to $1$ (that is,
it is a $\ONEONELQ$ instance) and it admits a feasible, stable matching $M_s = \{(r_1^i,h_3^i), (r_2^i,h_2^i) \mid v_i \text{ is a vertex in } G\}$. Thus, $d = n_d = 0$.
If $k'$ is either $d$ or $n_d$ and $G'$ is the reduced instance of $\MAXRSM$
	constructed in~\cite{DBLP:SAGT20}, then
Theorem~4~\cite{DBLP:SAGT20} implies that $\MAXRSM$ is $\NP$-hard when the chosen parameter is constant.
This implies that $\MAXRSM$ is para-$\NP$-hard when parameterized by $d$ or $n_d$ even for $\ONEONELQ$ instances.

\section{Concluding Remarks}\label{sec:disc}
An instance with lower-quotas may not admit a stable, feasible matching and therefore
a relaxation of stability is investigated.
Various optimality notions for lower-quota setting are studied~\cite{HIM16, Yokoi20, FITUY15, NN17} and
parameterized complexity for the notion studied in~\cite{HIM16} is investigated in~\cite{DBLP:journals/algorithmica/MnichS20}.
In this work, we investigate the notion of envy-freeness and relaxed stability
in presence of lower-quotas.
On one hand, envy-freeness is a natural relaxation of stability that respects
the preferences of agents and resources to the {\em best possible extent} but an instance may not
admit a feasible, envy-free matching.
On the other hand, an instance always admits a feasible, relaxed stable matching.
Krishnaa~et~al.~\cite{DBLP:SAGT20} show that computing a maximum size feasible,
envy-free or relaxed stable matching is computationally hard.
Prior to this work, no results about parameterized complexity of computing a largest feasible envy-free or
relaxed stable matching were known.

We present our kernelization and $\FPT$ results for $\ONEONELQ$ instances.
These results are theoretically challenging as well as practically appealing.
A standard technique of {\em cloning}~\cite{GI89, DBLP:journals/algorithmica/MnichS20} 
can be used to efficiently convert a $\MANYONELQ$ instance to a $\ONEONELQ$.
We observe that the parameters $s, |\overline{\A}|, \ell_{LQ}$
remain the same after cloning.
Thus, our kernelization for $\MAXRSM$ (in Theorem~\ref{thm:kernel}) and
$\FPT$ result for the parameter $|\overline{\A}|$ (in Theorem~\ref{thm:fpt})
are also applicable for the $\MANYONELQ$ instances.
We also remark that algorithm $\ALG$ can be appropriately modified
to handle arbitrary quotas in the $\MANYONELQ$ instances. 
Using the modified $\ALG$ algorithm, $\alg_{efm}$ 
can be applied to $\MANYONELQ$ instances {\em without} using the
cloning technique for the parameters $(q, \ell_{LQ})$.
In that case, the running time of $\alg_{efm}$ changes to $O(m\cdot 2^{q\ell_{LQ}})$.
Thus, the new algorithm remains $\FPT$.

\vspace{0.2cm}
\noindent{\bf Acknowledgements. } The author thanks anonymous reviewers for 
their comments that provided technical insights and significantly improved the presentation of the paper.

\bibliographystyle{splncs04}
\bibliography{refs}

\appendix
\end{document}